\documentclass{article}

\usepackage{microtype}
\usepackage{graphicx}
\usepackage{subfigure}
\usepackage{booktabs} 
\usepackage{nicefrac}
\usepackage{hyperref}
\usepackage{multicol}
\usepackage{wrapfig}

\usepackage{algorithm,algpseudocode}

\usepackage[nonatbib,preprint]{neurips_2021}

\usepackage[style=alphabetic,backend=bibtex,sorting=anyt]{biblatex}
\usepackage[utf8]{inputenc} 
\usepackage[T1]{fontenc}    
\RequirePackage{amsthm,amsmath,amsfonts,amssymb}

\usepackage{graphicx}
\usepackage{enumerate}
\usepackage{verbatim}
\usepackage{mathrsfs}
\usepackage{paralist}
\usepackage{pstool}
\usepackage{booktabs}
\usepackage{tikz}
\usepackage{mathtools}
\usepackage{bbm}
\usepackage{url}

\newtheorem{Assumption}{Assumption}

\newtheorem{Lemma}{Lemma}
\newtheorem{Theorem}{Theorem}
\newtheorem{Corollary}{Corollary}

\newtheorem{Definition}{Definition}

\DeclareMathOperator{\pd}{\partial}

\DeclareMathOperator{\vep}{\varepsilon}

\DeclareMathOperator{\rank}{rank}

\newcommand{\pN}{\mathcal{N}}
\newcommand{\mv}[1]{{\boldsymbol{\mathrm{#1}}}}

\newcommand{\A}{\mathscr{C}}
\newcommand{\Ac}{\mathscr{U}}

\newcommand{\csolve}[1]{\mathcal{S}_{#1}}
\newcommand{\cchol}[1]{\mathcal{C}_{#1}}

\graphicspath{{../tex/figs/}}

\title{Efficient methods for Gaussian Markov random fields under sparse linear constraints}
\author{%
	David~Bolin\\
	King Abdullah University of \\Science and Technology\\
	\texttt{david.bolin@kaust.edu.sa}\\
	\And
	Jonas~Wallin\\
	Department of Statistics,\\
	 Lund University\\
	\texttt{jonas.wallin@stat.lu.se}\\
}

\begin{document}

\maketitle

\vskip 0.3in

\begin{abstract}
%
%
%
%
%

Methods for inference and simulation of linearly constrained Gaussian Markov Random Fields (GMRF) are computationally prohibitive when the number of constraints is large. In some cases, such as for intrinsic GMRFs, they may even be unfeasible. We propose a new class of methods to overcome these challenges in the common case of sparse constraints, where one has a large number of constraints and each only involves a few elements. Our methods rely on a basis transformation into blocks of constrained versus non-constrained subspaces, and we show that the methods greatly outperform existing alternatives in terms of computational cost. By combining the proposed methods with the stochastic partial differential equation approach for Gaussian random fields, we also show how to formulate Gaussian process regression with linear constraints in a GMRF setting to reduce computational cost. This is illustrated in two applications with simulated data.

\end{abstract}

\section{Introduction}

Linearly constrained Gaussian processes have recently gained attention, especially for Gaussian process regression where the model should obey some underlying physical principle such as conservation laws or equilibrium conditions  \cite{NIPS2010_4179,Sarkka2011,Wahlstrom2013,NIPS2017_6721,NIPS2018_7483,Solin2018,Hegermann2021}. A well-known challenge with these models, and Gaussian processes in general, is their high computational cost for inference and prediction in the case of big data sets \cite{6884588,Cong2017hyper,Jidling2018}. 
One way to reduce computational burden is to impose conditional independence assumptions. In fact, conditional independence between random variables is often explicitly or implicitly assumed in large classes of statistical models including Markov processes, hierarchical models, and graphical models. The assumption typically increases the model's interpretability and facilitates computationally efficient methods for inference \cite{rue09}.
Gaussian variables with conditional independence properties are known as Gaussian Markov random fields (GMRFs), and these are widely used in areas ranging from image analysis \cite{penny2005bayesian} to spatial statistics \cite{Bolin2009} and time series analysis \cite{rue2005gaussian}. GMRFs also arise as computationally efficient approximations of certain Gaussian processes  \cite{lindgren11}, which is a fundamental modeling tool in both machine learning and statistics. 
In particular, such approximations in combination with the integrated nested Laplace approximation (INLA) methodology \cite{rue09} made latent GMRFs widely used in the applied sciences \cite{Bakka2018}. GMRFs also have  connections with convolutional neural networks leading to recent considerations of  deep GMRFs~\cite{siden2020}.

The key feature of GMRFs that reduces computational cost is sparsity. Specifically, a GMRF
\begin{equation}\label{eq:distX}
	\mv{X} = [X_1,\ldots,X_n]^\top \sim \mathcal{N}\left( \mv{\mu}, \mv{Q}^{-1} \right),
\end{equation}
has a sparse precision (inverse covariance) matrix $\mv{Q}$ which  enables the use of sparse matrix techniques for computationally efficient sampling and statistical inference.
The sparsity is caused by conditional independence assumptions since $Q_{ij} = 0$ if and only if the two variables $X_i$ and $X_j$ are  independent conditionally on all other variables in $\mv{X}$ \cite[see][Chapter 2]{rue2005gaussian}. 

For a GMRF $\mv{X}$, a set of $k$ linear constraints can be formulated as
\begin{equation}\label{eq:contraint}
	\mv{A}\mv{X} = \mv{b},
\end{equation}
where each row in the $k\times n$ matrix $\mv{A}$ and the vector $\mv{b}$ encodes a constraint on $\mv{X}$.  For example, a sum-to-zero constraint $\sum_{i=1}^n X_i = 0$ can be written in this way, which is commonly used in hierarchical models to ensure identifiability.  Observations of GMRFs can also be formulated as linear constraints, where, e.g.,  a point observation of $X_i$ is  the simple constraint $X_i=x_i$. 
These  deterministic restrictions are often referred to as hard constraints. If \eqref{eq:contraint} is assumed to hold up to Gaussian noise, i.e., $\mv{A}\mv{X} \sim \pN(\mv{b},\sigma^2\mv{I})$, then the constraints are instead referred to as soft constraints. This scenario is common when GMRFs are incorporated in hierarchical models, where the soft constraints represent noisy observations. 

The problem with adding hard constraints to GMRFs is that it can remove the computational advantages of the non-constrained model. 
Specifically, current methods for GMRFs with hard constraints have a computational cost that scales cubically in the number of constraints. 
This will thus be prohibitive when there are many constraints, which for example is common for constrained Gaussian processes like those considered by \textcite{NIPS2017_6721}.  

The code for reproducing all results are available at can be found at \url{https://github.com/JonasWallin/CB/} 

{\bf Summary of contributions.}
The main contribution of this work is the formulation of a novel class of computationally efficient methods for 
linearly constrained GMRFs in tasks such as parameter estimation and simulation. 
These methods  explore the conditional independence structure to reduce computational costs 
compared to traditional methods for situations with large numbers of constraints. 
The focus is in particular on the case, referred to as sparse hard constraints, when each constraint only involves a few number of variables so that $\mv{A}$ is sparse.
For this case, the main idea is to perform a change of basis so that 
the constraints are simpler to enforce in the transformed basis. 
In order to use these models for Gaussian process regression, 
the methods are also generalized to GMRF models with both hard and soft constraints. 
An important feature of the new class of methods, that previous approaches lack, 
is its applicability to intrinsic GMRFs,  which are improper in the sense that a set of eigenvalues of $\mv{Q}$ is zero. 
This makes their distributions invariant to shifts in certain directions, 
which is a useful property for prior distributions of Bayesian models in areas such as medical image analysis \cite{penny2005bayesian} and geostatistics \cite{Bolin2009}. 
The final contribution is the derivation of GMRFs for constrained Gaussian processes, 
by combining the proposed methods with the stochastic partial differential equation (SPDE) approach by \cite{lindgren11} and the nested SPDE methods by \cite{bolin11}. The combined approach is highly computationally efficient compared to standard covariance-based methods, as illustrated in two simulation studies.

{\bf Outline.} In Section~\ref{sec:old_methods}, the problem is introduced in more detail and the most commonly used methods for sampling and likelihood computations for GMRFs are reviewed.  Section~\ref{sec:new_method} introduces the new methods for GMRFs with sparse hard constraints. 
These methods are extended to the case with both hard and soft constraints  in Section~\ref{sec:extension}, followed by the  GMRF methods for constrained Gaussian processes in Section~\ref{sec:constrainedGP}. The methods are illustrated nuumerically in Section~\ref{sec:illustrations}. 
A discussion closes the article, which is supported by three appendices containing proofs and technical details.

\section{Standard methods for GMRFs under hard constraints}\label{sec:old_methods}
Hard  constraints can be divided into interacting and non-interacting constraints. In the latter, \eqref{eq:contraint}  specifies a constraint on a subset of the variables in $\mv{X}$ so that $\mv{AX} = \mv{b}$ can be written as $\mv{X}_c = \mv{b}_c$, where $c$ denotes a subset of the variables. Specifically, let $\mv{X}_u$ denote the remaining variables, then
$$
\mv{X} =  \begin{bmatrix}
	\mv{X}_c\\
	\mv{X}_{u}
\end{bmatrix}	\sim \mathcal{N}\left(
\begin{bmatrix}
	\mv{\mu}_c\\
	\mv{\mu}_{u}
\end{bmatrix}	
,
\begin{bmatrix}
	\mv{Q}_{cc} & \mv{Q}_{cu}\\
	\mv{Q}_{uc} & \mv{Q}_{uu}
\end{bmatrix}^{-1}	
\right) \quad
\mbox{and} \quad
\mv{X}|\mv{AX} = \mv{b} \sim \mathcal{N}\left(
\begin{bmatrix}
	\mv{b}_c\\
	\mv{\mu}_{u|c}
\end{bmatrix}	
,
\begin{bmatrix}
	\mv{0} & \mv{0}\\
	\mv{0} & \mv{Q}_{uu}^{-1}
\end{bmatrix}	
\right),
$$
where $\mv{\mu}_{u|c} = \mv{\mu}_u-\mv{Q}_{uu}^{-1}\mv{Q}_{uc}(\mv{b_c}-\mv{\mu}_c)$. Thus, we can split the variables into two subsets and treat the unconditioned variables separately.
The more difficult and interesting situation is the case of  interacting hard constraints where a simple split of the variables is not possible. From now on, we will assume that we are in this scenario. 

Our aim is to construct methods for sampling  from the distribution of $\mv{X}| \mv{AX} = \mv{b}$ and for evaluating its log-likelihood function. It is straightforward to show that $\mv{X}| \mv{AX} = \mv{b}\sim \mathcal{N}(\widehat{\mv{\mu}},\widehat{\mv{\Sigma}})$, where 
$\widehat{\mv{\mu}} = \mv{\mu} - \mv{Q}^{-1}\mv{A}^\top(\mv{A}\mv{Q}^{-1}\mv{A}^\top)^{-1}(\mv{A}\mv{\mu}-\mv{b})$ and
$\widehat{\mv{\Sigma}} = \mv{Q}^{-1} - \mv{Q}^{-1}\mv{A}^\top(\mv{A}\mv{Q}^{-1}\mv{A}^\top)^{-1}\mv{A}\mv{Q}^{-1}$.
Since $\widehat{\mv{\Sigma}}$ has rank $n-k$, likelihood evaluation and sampling is in general expensive. 
For example, one way is to use an eigenvalue decomposition of $\widehat{\mv{\Sigma}}$ \cite[see][Chapter 2.3.3]{rue2005gaussian}. 
However, this procedure  is not practical since it cannot take advantage of the sparsity of $\mv{Q}$.
Also, for intrinsic GMRFs, $\widehat{\mv{\Sigma}}$ and $\widehat{\mv{\mu}}$ cannot be constructed through the expressions above since $\mv{Q}^{-1}$ is unbounded. 

A commonly used method for sampling under hard linear constraints, sometimes referred to as conditioning by kriging \cite{rue2001sampling}, is to first sample $\mv{X} \sim \mathcal{N}\left( \mv{\mu}, \mv{Q}^{-1} \right)$ and then correct for the constraints by using
	$\mv{X}^* = \mv{X} - \mv{Q}^{-1}\mv{A}^\top(\mv{A}\mv{Q}^{-1}\mv{A}^\top)^{-1}(\mv{AX}- \mv{b})$
as a sample from the conditional distribution. 
Here the cost of sampling $\mv{X}$ is $\cchol{\mv{Q}} + \csolve{\mv{Q}}$, where $\cchol{\mv{Q}}$ denotes the computational cost of a sparse Cholesky factorization $\mv{Q} = \mv{R}^\top\mv{R}$ and $\csolve{\mv{Q}}$ the cost of solving $\mv{Rx} = \mv{u}$ for  $\mv{x}$ given $\mv{u}$. Adding the cost for the correction step,  the total cost of the method is 
$
\mathcal{O}(\cchol{\mv{Q}} + (k+2)\csolve{\mv{Q}} + k^3).
$

Let  $\pi_{\mv{Ax}}(\cdot)$ denote the density of  $\mv{Ax}\sim\pN(\mv{A\mu},\mv{A}\mv{Q}^{-1}\mv{A}^\top)$, then  
the likelihood of $\mv{X}|\mv{AX}=\mv{b}$ can be computed through the expression 
\begin{equation}\label{eq:old_likelihood}
	\pi(\mv{x}|\mv{Ax}=\mv{b}) = \frac{\pi_{\mv{Ax}|\mv{x}}^*(\mv{b}|\mv{x})\pi(\mv{x})}{\pi_{\mv{Ax}}(\mv{b})}, 
\end{equation}
where $\pi_{\mv{Ax}|\mv{x}}^*(\mv{b}|\mv{x})=\mathbb{I}(\mv{Ax}=\mv{b})|\mv{A}\mv{A}^\top|^{-1/2}$ and $\mathbb{I}(\mv{Ax}=\mv{b})$ denotes the indicator function with $\mathbb{I}(\mv{Ax}=\mv{b}) = 1$ if $\mv{Ax}=\mv{b}$  and $\mathbb{I}(\mv{Ax}=\mv{b}) = 0$ otherwise. 
This result is formulated in \cite{rue2001sampling} and we provide further details in Appendix~\ref{app:conditional} by showing that \eqref{eq:old_likelihood} is a density with respect to the Lebesgue measure on the level set $\{\mv{x}:\mv{Ax}=\mv{b}\}$.
%
The computational cost of evaluating the likelihood using this formulation is 
$
\mathcal{O}(\cchol{\mv{Q}} + (k+1)\csolve{\mv{Q}} + k^3).
$

 Note that these methods only work efficiently for a small number of constraints, because of the term $k^3$ in the computational costs, and for proper GMRFs. In the case of intrinsic GMRFs we cannot work with $\mv{A}\mv{Q}^{-1}\mv{A}$ due to the rank deficiency of $\mv{Q}$. 

\section{The basis transformation method}\label{sec:new_method}
In this section, we propose the new class of methods in two steps. We first derive a change of basis in Section~\ref{sec:basis}, and then use this to formulate the desired conditional distributions in Section~\ref{sec:sampling_new}. The resulting computational costs of  likelihood evaluations and simulation are discussed in Section~\ref{sec:usage}. 

Before stating the results we introduce some basic notation. 
When working with intrinsic GMRFs the definition in \eqref{eq:distX} is inconvenient since the covariance matrix has infinite eigenvalues. 
Instead one can use the canonical parametrization
$
\mv{X} \sim \mathcal{N}_{C}\left(\mv{\mu}_C, \mv{Q} \right), 
$ 
which implies that the density of $\mv{X}$ is given by
$
\pi(\mv{x}) \propto \exp \left(- \frac12\mv{x}^\top\mv{Q}\mv{x} + \mv{\mu}_C^\top\mv{x} \right)
$
and thus that $\mv{\mu}_C = \mv{Q}\mv{\mu}$.
Also, since we will be working with non-invertible matrices we will need the Moore--Penrose inverse and the pseudo determinant. 
We denote the Moore--Penrose inverse of a matrix $\mv{B}$ by $\mv{B}^{\dagger}$ and for a symmetric positive semi definite matrix $\mv{M}$ we define the  pseudo determinant as $|\mv{M}|^\dagger = \prod_{i:\lambda_i>0}\lambda_i$ where $\{\lambda_i\}$ are the eigenvalues of $\mv{M}$. Finally, for the remainder of this article, we will assume the following. 

\begin{Assumption}
	\label{ass:main}
	 $\mv{X} \sim \mathcal{N}_C\left(\mv{Q}\mv{\mu},\mv{Q}\right)$ where $\mv{Q}$ is a  positive semi-definite $n\times n$ matrix with rank $n-s>0$ and null-space $\mv{E}_Q$.  $\mv{A}$ is a $k\times n$ matrix with rank $k$ and $rank(\mv{AE}_Q)=k_0$.
\end{Assumption}


\subsection{Basis construction}\label{sec:basis}

Our main idea is to construct a basis on $\mathbb{R}^n$ such that the constraints are easily enforced. 
A key property of the basis is that $\mv{A}$ should be spanned by the first $k$ elements of the basis. 
Essentially, this means that we are transforming the natural basis into one where the results for the case of non-interacting hard constraints can be used. 
The basis is defined by an $n\times n$ change-of-basis matrix which we denote $\mv{T}$. 
In Algorithm~\ref{alg:1} we present a simple method to produce such a matrix, using the singular value decomposition (SVD). 
In the algorithm, $id(\mv{A})$ is a function that returns the indices of the non-zero columns in $\mv{A}$ and $\mv{A}_{\mv{\cdot},D}$ denotes the matrix obtained by extracting the columns in $\mv{A}$ with indices in the set $D$. 
The computational cost of the method is dominated by that of the SVD, which is 
$\mathcal{O}\left(k^3 + k^2|id(\mv{A})| \right)$  \cite[][p.~493]{golub2013matrix}.

\begin{wrapfigure}{r}{0.5\textwidth}
	\begin{minipage}{0.5\textwidth}
		\begin{algorithm}[H]
			\caption{Constraint basis construction.}
			\label{alg:1}
			\begin{algorithmic}[1]
				\Require $\mv{A}$ (a $k \times n$ matrix of rank $k$)
				\State $\mv{T}  \gets \mv{I}_n$ 
				\State $D \gets  id(\mv{A})$     
				\State $\mv{USV}^\top \gets svd(\mv{A}_{\mv{\cdot},D})$ 
				\State $\mv{T}_{D,D} \gets \mv{V} ^\top$
				\State  $\mv{T} \gets \begin{bmatrix}
					\mv{T}_{\mv{\cdot},D} &
					\mv{T}_{\mv{\cdot},D^c}
				\end{bmatrix}$
				\State Return $\mv{T}$
			\end{algorithmic}
		\end{algorithm}	
	\end{minipage}
\end{wrapfigure}
Clearly, the cubic scaling in the number of constraints may reduce the efficiency of any method that requires this basis construction as a first step. 
However, suppose that the rows of $\mv{A}$ can be split into two sub-matrices, $\widetilde{\mv{A}}_1$ and $\widetilde{\mv{A}}_2$, which have no common non-zero columns. 
Then the SVD of the two matrices can be computed separately. Suppose now that $\mv{A}$ corresponds to $m$ such sub-matrices, then, after reordering, $\mv{A} = \bigl[\widetilde{\mv{A}}_1^\top,\ldots, \widetilde{\mv{A}}_m^\top \bigr]^\top$ where 
$\{\tilde{\mv{A}}_i\}_{i=1}^m$ represent sub-constraints  such that $id(\tilde{\mv{A}}_i)\cap id(\tilde{\mv{A}}_l)= \emptyset$ for all $i$ and $l$. By replacing the SVD of Algorithm~\ref{alg:1} by the $m$ SVDs of the  lower-dimensional matrices the computational cost is reduced to
$
\mathcal{O}\left( \sum_{i=1}^m rank(\tilde{\mv{A}}_i)^3 + rank(\tilde{\mv{A}}_i)^2|id(\tilde{\mv{A}}_i)|\right)$.
This method is presented in Algorithm~\ref{alg:2}. The reordering step is easy to perform and is described  in Appendix~\ref{sec:details}, where also more details about the algorithm are given.
\begin{algorithm}[b]
	\caption{$CB(\mv{A})$.  Constraint basis construction for non-overlapping subsets of constraints. }
	\label{alg:2}
	\begin{multicols}{2}
		\begin{algorithmic}[1]
			\Require $\mv{A}$ (a $k\times n$ matrix of rank $k$)
			\State $\mv{T}  \gets \mv{I}_n$ 
			\State $D_{full} \gets  id(\mv{A})$ 
			\State $h\gets 1$ 
			\State $l\gets k+1$ 
			\State Reorder so that $\mv{A} = \begin{bmatrix}\widetilde{\mv{A}}_1^\top & \ldots & \widetilde{\mv{A}}_m^\top\end{bmatrix}^\top$
			\For{$i=1:m$ }
			\State $D \gets  id(\tilde{\mv{A}}_i)$ 
			\State $\mv{USV}^\top \gets svd(\widetilde{\mv{A}}_{i})$ 
			\State $u \gets ncol(\mv{U})$
			\State $\mv{T}_{h:(h+u),D} \gets \bigl(\mv{V}^\top\bigr)_{1:u,D} $
			\State $h \gets h+u+1$
			\If{$|D|>u$}
			\State $\mv{T}_{l:(l+|D|-u),D} \gets \bigl(\mv{V}^\top\bigr)_{(u+1):|D|,D} $
			\State $l \gets l+|D|-u+1$
			\EndIf
			\EndFor
			\State   $\mv{T} \gets \begin{bmatrix}
				\mv{T}_{\mv{\cdot},1:|D_{full}|} &
				\left(\mv{I}_n\right)_{\mv{\cdot},D_{full}^c}
			\end{bmatrix}$
			\State Return $\mv{T}$
		\end{algorithmic}
	\end{multicols}
\end{algorithm}

\subsection{Conditional distributions}\label{sec:sampling_new}
Using the change of basis from the previous subsection, we can now derive alternative formulations of the distributions of $\mv{AX}$ and $\mv{X}|\mv{AX}=\mv{b}$ which are suitable for sampling and likelihood-evaluation. There are two main results in this section. The first provides an expression of the density of $\mv{AX}$ that allows for computationally efficient likelihood evaluations for observations $\mv{AX} = \mv{b}$. The second formulates the conditional distribution for $\mv{X}|\mv{AX}=\mv{b}$ in a way that allows for efficient sampling of $\mv{X}$ given observations $\mv{AX}=\mv{b}$. 
To formulate the results, let $\mv{T}=CB(\mv{A})$ be the output of Algorithm~\ref{alg:1} or Algorithm~\ref{alg:2} and $\mv{X}^*=\mv{T}\mv{X}$ which, under Assumption \ref{ass:main}, has distribution 
$
\mv{X}^* \sim \mathcal{N}_C \left(\mv{Q}^*\mv{\mu}^*,\mv{Q}^*\right),
$
where
$\mv{\mu}^*=\mv{\mv{T}\mu}$ and $ \mv{Q}^* = \mv{T} \mv{Q}\mv{T}^\top$.
Henceforth we use stars to denote quantities such as means and precisions in the transformed space.  Note that we can move from the transformed space to the original space by multiplying with $\mv{T}^\top$ for a vector and by multiplying with $\mv{T}^\top$ from the left and $\mv{T}$ from the right for a matrix.

Since $\mv{A}$ is spanned by the first $k$ elements in $\mv{T}$, we use the index notation $\A=\{1,\ldots,k\}$ and $\Ac=\{k+1,\ldots,n\}$. We also use the matrix $\mv{H} = \bigl(\mv{A}\mv{T}^\top\bigr)_{\A\A}$, which is equal to $\mv{US}_{\A\A}$ from the SVD in Algorithm \ref{alg:1}, and therefore has inverse $\mv{H}^{-1}= \mv{S}^{-1}_{\A\A}\mv{U}^\top$. Finally we define $\mv{b}^*= \mv{H}^{-1}\mv{b}$.
%
\begin{Theorem}
	\label{Them:piAX}
	Under Assumption \ref{ass:main} it follows that
	\begin{align*}
		\pi_{\mv{AX}}(\mv{b}) =  \frac{ 
			| \mv{Q}_{\A|\Ac}^*|^{\frac{\dagger}{2}} }{(2\pi)^{k/2}|\mv{A}\mv{A}^\top|^{1/2}} \cdot
		\exp \left(- \frac{1}{2}\left( 	\mv{b}^* - \mv{\mu}^*_\A\right)^\top \mv{Q}_{\A|\Ac}^* \left(	\mv{b}^* - \mv{\mu}_{\A}^*\right) \right),
	\end{align*}
	where 
	$
\mv{Q}_{\A|\Ac}^* = \mv{Q}_{\A\A}^{*} -  \mv{Q}_{\A\Ac}^{*}\left(\mv{Q}_{\Ac\Ac}^{*}\right)^{\dagger}\mv{Q}_{\Ac\A}^{*}
	$ and $|\mv{Q}_{\A|\Ac}^*|^{\frac{\dagger}{2}}  = |\mv{Q}|^{\frac{1}{2}} |\mv{Q}^*_{\Ac\Ac}|^{-\frac{1}{2}}$.
\end{Theorem}
If $\mv{Q}$ is positive definite we have 
$\mv{Q}_{\A|\Ac}^* = \mv{Q}_{\A\A}^{*} -  \mv{Q}_{\A\Ac}^{*}\left(\mv{Q}_{\Ac\Ac}^{*}\right)^{-1}\mv{Q}_{\Ac\A}^{*}$ and we can then
replace $| \mv{Q}_{\A|\Ac}^*|^{\frac{\dagger}{2}}$ with $| \mv{Q}_{\A|\Ac}^*|^{\frac{1}{2}}$ in the expression of $\pi_{\mv{A}\mv{X}}$.
%

\begin{Theorem}
	\label{thm:piXAX}
	Under Assumption \ref{ass:main}  it follows that
	\begin{align}
		\label{eq:densXAXb}
		\mv{X}|\mv{A}\mv{X}=\mv{b} \sim \mathcal{N}_C\left(\mv{Q}_{X|b}\widetilde{\mv{\mu}},\mv{Q}_{X|b} \right)\mathbb{I}\left(\mv{AX}=\mv{b}\right),
	\end{align}
	where $	\mv{Q}_{X|b} = \mv{T}_{\Ac,}^\top \mv{Q}_{\Ac\Ac}^*\mv{T}_{\Ac,}$ is positive semi-definite with rank $n-s-(k-k_0)$ and $\widetilde{\mv{\mu}} =    \mv{T}^\top 	\widetilde{\mv{\mu}}^*$ with
$
		\widetilde{\mv{\mu}}^* = 
			\scalebox{0.75}{$\begin{bmatrix}
		\mv{b}^* \\
			\mv{\mu}_{\Ac}^* - \mv{Q}_{\Ac\Ac}^{*\dagger}  \mv{Q}^*_{\Ac\A} \left( \mv{b}^*  - \mv{\mu}^*_{\A} \right)
		\end{bmatrix}$}$.
\end{Theorem}

Note that $\mv{Q}_{X|b}\mv{A}=\mv{0}$, which implies that the right side of \eqref{eq:densXAXb} is a (possibly intrinsic) density with respect to Lebesgue measure on the level set $\{\mv{x}:\mv{Ax}=\mv{b}\}$. Further, note that  $\mv{Q}_{\Ac\Ac}^*\mv{T}_{\Ac}\mv{E}_0=\mv{0}$, which implies that $\mv{X}$ is improper on the span of $\mv{T}_{\Ac}\mv{E}_0$.

If $\mv{Q}$ is positive definite we get the following corollary.

\begin{Corollary}
	\label{col:simple}
	Under Assumption \ref{ass:main} with $s=0$, we have
	\begin{align}
		\label{eq:densXAXb2}
		\mv{X}|\mv{A}\mv{X}=\mv{b} &\sim \mathcal{N}\left(
		\widetilde{ \mv{\mu}}, \widetilde{\mv{\Sigma}}\right)\mathbb{I}\left(\mv{AX}=\mv{b}\right),
	\end{align}
	where $	\widetilde{\mv{\Sigma}} = \mv{T}_{\Ac,}^\top \left(\mv{Q}_{\Ac\Ac}^* \right)^{-1}\mv{T}_{\Ac,}$ is a positive semi-definite matrix of rank $n-k$ and $\widetilde{\mv{\mu}} =  \mv{T}^\top \widetilde{\mv{\mu}}^*$ with
$
\widetilde{\mv{\mu}}^* =  
			\scalebox{0.75}{$\begin{bmatrix}
			\mv{b}^* \\
			\mv{\mu}_{\Ac}^* - \left( \mv{Q}_{\Ac\Ac}^{*} \right)^{-1} \mv{Q}^*_{\Ac\A} \left( \mv{b}^*  - \mv{\mu}^*_{\A} \right)
		\end{bmatrix}.$}
$
\end{Corollary}

\subsection{Sampling and likelihood evaluations}\label{sec:usage}
The standard method for sampling a GMRF  $\mv{X} \sim \mathcal{N}\left(\mv{\mu},\mv{Q}^{-1}\right)$ is to first compute the Cholesky factor $\mv{R}$ of $\mv{Q}$, then sample $\mv{Z}\sim \mathcal{N}\left(\mv{0},\mv{I} \right)$, and finally set
\begin{align}
	\label{eq:sparsesample}	
	\mv{X} = \mv{\mu} +\mv{R}^{-1}\mv{Z}.
\end{align}
To sample $\mv{X}|\mv{AX}=\mv{b}$ we use this method in combination with Theorem~\ref{thm:piXAX} as shown in Algorithm~\ref{alg:sampling}.  The cost of using the algorithm for sampling, and for computing the expectation of $\mv{X}$ in Theorem~\ref{thm:piXAX}, is dominated by $\cchol{\mv{Q}^*_{\Ac\Ac}}$ given that $\mv{T}$ has been pre-computed.  
Similarly, the cost for evaluating the likelihood in Theorem \ref{Them:piAX} is dominated by the costs of the Cholesky factors $\cchol{\mv{Q}^*_{\Ac\Ac}}  +\cchol{\mv{Q}} + \cchol{\mv{A}\mv{A}^\top}$. 

These costs are not directly comparable to costs of the methods from Section~\ref{sec:old_methods} since they involve operations with the transformed precision matrix $\mv{Q}^*=\mv{T}\mv{Q}\mv{T}^\top$ which may have a different, and often denser, sparsity structure than $\mv{Q}$. In fact if $\mv{T}$ is dense the method will not be practically useful since even the construction of $\mv{Q}^*$ would be $\mathcal{O}\left(n^2\right)$. Thus, to understand the computational cost we must understand the sparsity structure of the transformed matrix. To that end, first note that only the rows $id(\mv{A})$ in $\mv{Q}^*$ will have a sparsity structure that is different from that in $\mv{Q}$. In general, the variables involved for the $i$th constraint,  $id(\mv{A}_i\mv{X})$, will in the constrained distribution share all their neighbors. This implies that if $i\in id(\mv{A})$, then $|\mv{Q}_{i,j}^*|>0$ if 
$|\mv{Q}_{i,j}|>0$ and we might have $|\mv{Q}_{i,j}^*|>0$ if $\sum_{k\in {id(\mv{A}})}|\mv{Q}_{k,j}^*|>0$. This provides a worst-case scenario for the amount of non-zero elements in $\mv{Q}^*$, where we see that the sparsity of the constraints is important. 

\section{GMRFs under hard and soft constraints}\label{sec:extension}
As previously mentioned, one can view observations of a GMRF as hard constraints. In many cases, these observations are assumed to be taken under Gaussian measurement noise, which can be seen as soft constraints on the GMRF. It is therefore common to have models with both soft and hard constraints (e.g., a model with noisy observations of a field with a sum-to-zero constraint). Here, we extend the methods of the previous section to this case. 
Specifically, we consider the following  hierarchical model
\begin{equation}\label{eq:hierarchical}
	\begin{split}
\mv{X}& \sim \mathcal{N}_C\left(\mv{Q}\mv{\mu},\mv{Q}  \right), \quad \mbox{ subject to }
\mv{AX} = \mv{b}, \\
\mv{Y} &\sim \mathcal{N}\left(\mv{B}\mv{X}, \sigma^2_Y\mv{I} \right),
\end{split}
\end{equation}
where  $\mv{Y}\in\mathbb{R}^m$ represent noisy observations of the linear combinations $\mv{BX}$ of $\mv{X} \in\mathbb{R}^n$, with $m\leq n$, satisfying Assumption~\ref{ass:main}, and $\mv{B}$ is an $m\times n$ matrix with rank $m$.
To deal with this type of models we present two results in this section. 
First, Theorem~\ref{Them:piAXsoft} shows how to compute the likelihood of the model. Second, the result in Theorem~\ref{Them:piXgby} can be used to efficiently compute the mean of $\mv{X}$ given the constraints and to sample from it.

We use the hat notation -- like $\widehat{\mv{Q}}$ -- to denote quantities for distributions conditionally on the observations $\mv{Y}=\mv{y}$. We also use the notation from Theorem \ref{thm:piXAX} and additionally introduce $\mv{B}^*=\mv{BT}^\top$ and $\mv{y}^* = \mv{y}-\mv{B}\mv{T}_\A^\top \mv{b}^*$.
We start by deriving the likelihood, $\pi_{\mv{Y}|\mv{AX}} (\mv{y}|\mv{b})$, which is needed for inference.

\begin{Theorem}
	\label{Them:piAXsoft}
	For the model in \eqref{eq:hierarchical} one has
	\begin{align*}
	\pi_{\mv{Y}|\mv{AX}}(\mv{y}|\mv{b}) = &
	\frac{\sigma_Y^{-m}|\mv{Q}^*_{\Ac\Ac}|^{\frac{\dagger}{2}} }{\left(2\pi \right)^{c_0} |\widehat{\mv{Q}}^*_{\Ac\Ac} |^{\frac{\dagger}{2}}}  
	 \exp \left(-  \frac{1}{2}\left[\frac{ \mv{y}^{*T}\mv{y}^*}{\sigma^2_Y}+ \widetilde{\mv{\mu}}_{\Ac} ^{*\top} \mv{Q}^*_{\Ac\Ac}\widetilde{\mv{\mu}}^*_{\Ac}- \widehat{\mv{\mu}}_{\Ac}^{*\top}\widehat{\mv{Q}} ^*_{\Ac\Ac}  \widehat{\mv{\mu}}^*_{\Ac}\right]\right),
	\end{align*}
	where $c_0>0$, and 
\begin{align*}
\widehat{\mv{Q}}^*_{\Ac\Ac} &= \mv{Q}^*_{\Ac\Ac}  +\frac{1}{\sigma^2_Y} \left(\mv{B}^*_{\Ac}\right)^\top  \mv{B}^*_{\Ac},  \\
\quad \widehat{\mv{\mu}}^*_{\Ac} &=\widehat{\mv{Q}}^{*\dagger}_{\Ac\Ac}  \left(
\mv{Q}^*_{\Ac\Ac} \widetilde{\mv{\mu}}^*_{\Ac} + \frac{1}{\sigma^2_Y}\left(\mv{B}^*_{\Ac}\right)^\top \mv{y}^* \right).
\end{align*}
\end{Theorem}

The computational cost of evaluating the likelihood is $\cchol{\widehat{\mv{Q}}^*_{\Ac\Ac}} + \csolve{\widehat{\mv{Q}}^*_{\Ac\Ac}} +\cchol{\mv{Q}^*_{\Ac\Ac}}$.
The following theorem contains the distribution of  $\mv{X}$ given the event $\{\mv{AX}=\mv{b},\mv{Y}=\mv{y}\}$, which for example is needed when the model is used for prediction. 


\begin{Theorem}
	\label{Them:piXgby}
	For model in \eqref{eq:hierarchical} one has
	$\pi_{\mv{X}| \mv{AX},\mv{Y}}(\mv{b},\mv{y}) \sim  \mathcal{N}\left(
\widehat{\mv{\mu}},
	\widehat{\mv{Q}}
	\right) 
$
	where $\widehat{\mv{Q}} = \mv{T}^\top_{\Ac,}	\widehat{\mv{Q}}^*_{\Ac\Ac} \mv{T}_{\Ac,}$ and 
	$
	\widehat{\mv{\mu}} = \mv{T}^\top
	\scalebox{0.75}{$\begin{bmatrix}
	\mv{b}^* \\
	\widehat{\mv{\mu}}^*_{\Ac} 
	\end{bmatrix}$}
	$. Here 
	$\widehat{\mv{Q}}^*_{\Ac\Ac}$ and 
	$\widehat{\mv{\mu}}^*_{\Ac}$ are given in Theorem~\ref{Them:piAXsoft}.
	Further, let $\mv{E}_{\mv{Q}^*_{\Ac\Ac}}$ be the null space of $\mv{Q}^*_{\Ac\Ac}$ then $\rank(\widehat{\mv{Q}}) = n-s-(k-k_0)+rank(\mv{B}^*_{\Ac} \mv{E}_{\mv{Q}^*_{\Ac\Ac}})$.
\end{Theorem}

Since the distribution in the theorem is a normal distribution, we can sample from $\mv{X}$ given the event $\{\mv{AX}=\mv{b},\mv{Y}=\mv{y}\}$ using sparse Cholesky factorization as shown in Algorithm~\ref{alg:sampling2}.

\begin{figure}
	\begin{center}
		\begin{minipage}[t]{0.45\linewidth}
			\begin{algorithm}[H]
				\caption{Sampling $\mv{X} \sim \mathcal{N} \left(\mv{\mu},\mv{Q}^{-1}\right)$ subject to $\mv{AX}=\mv{b}$. \protect\phantom{$\mv{Y} \sim \mathcal{N}\left(\mv{B}\mv{X}, \sigma^2_Y\mv{I} \right)$} }
				\label{alg:sampling}
				\begin{algorithmic}[1]
					\Require $\mv{A}, \mv{b}, \mv{Q}, \mv{\mu}, \mv{T}$
					\State $\A \gets 1:nrow(\mv{A})$
					\State $\Ac \gets \left(nrow(\mv{A})+1\right):ncol(\mv{A})$
					\State $\mv{Q}^* \gets \mv{T}\mv{Q}\mv{T}^\top$
					\State $\mv{R} \gets  chol(\mv{Q}_{\Ac\Ac}^*)$    
					\State $\mv{b}^* \gets solve\left( \bigl(\mv{A}\mv{T}^\top\bigr)_{\A\A},\mv{b}\right)$
					\State $\mv{m}^* \gets 	\mv{\mu}_{\Ac}  -$\newline \hspace*{0.9cm}$solve(\mv{R}^\top, 
					\mv{Q}^*_{\Ac\A} \left(\mv{b}^*-\mv{T}_{\A} \mv{\mu}\right)) $
					\State Sample $\mv{Z} \sim \mathcal{N}\left(\mv{0},\mv{I}_{\Ac\Ac}\right)$
					\State $\mv{X}^* \gets \begin{bmatrix}
						\mv{b}^* , 
					 solve(\mv{R}, \mv{m}^* + \mv{Z})
					\end{bmatrix}^\top$
					\State $\mv{X} \gets \mv{T}^\top \mv{X}^*$
					\State Return $\mv{X}$
				\end{algorithmic}
			\end{algorithm}	
		\end{minipage}
		\begin{minipage}[t]{0.54\linewidth}	
			\begin{algorithm}[H]
			\caption{Sampling $\mv{X} \sim \mathcal{N} \left(\mv{\mu},\mv{Q}^{-1}\right)$ subject to $\mv{AX}=\mv{b}$ and $\mv{Y}=\mv{y}$ where $\mv{Y} \sim \mathcal{N}\left(\mv{B}\mv{X}, \sigma^2_Y\mv{I} \right)$.}
			\label{alg:sampling2}
			\begin{algorithmic}[1]
				\Require $\mv{A}, \mv{b}, \mv{Q}, \mv{\mu}, \mv{T},\mv{y},\mv{B}, \sigma^2_Y$
				\State $\A \gets 1:nrow(\mv{A})$
				\State $\Ac \gets \left(nrow(\mv{A})+1\right):ncol(\mv{A})$
				\State $\mv{Q}^* \gets \mv{T}\mv{Q}\mv{T}^\top$
				\State $\mv{B}^* \gets\mv{B}\mv{T}^\top$
				\State $\mv{R} \gets  chol(\mv{Q}_{\Ac\Ac}^* + \frac{1}{\sigma^2_Y}\left(\mv{B}^*\right)^\top \mv{B}^*)$    
				\State $\mv{b}^* \gets solve\left( \bigl(\mv{A}\mv{T}^\top\bigr)_{\A\A},\mv{b}\right)$
				\State $\mv{y}^* \gets \mv{y}-\mv{BT}^\top_{\A}\mv{b}^*$
				\State $\mv{m}^* \gets solve\left(\mv{R}^\top, \mv{Q}_{\Ac\Ac}^*\mv{T}_{\Ac}\mv{\mu}+ \right.$\newline 
				\hspace*{0.9cm}$\left.\frac{1}{\sigma^2_Y} \left(\mv{B}^* \right)^\top \mv{y}^* -\mv{Q}^*_{\Ac\A} \left(\mv{b}^*-\mv{T}_{\A} \mv{\mu}\right)\right) $
				\State Sample $\mv{Z} \sim \mathcal{N}\left(\mv{0},\mv{I}_{\Ac\Ac}\right)$
				\State $\mv{X}^* \gets \begin{bmatrix}
					\mv{b}^* , 
					solve(\mv{R}, \mv{m}^* + \mv{Z})
				\end{bmatrix}^\top$
				\State $\mv{X} \gets \mv{T}^\top \mv{X}^*$
				\State Return $\mv{X}$
			\end{algorithmic}
		\end{algorithm}	
		\end{minipage}
	\end{center}
\end{figure}

\section{Constrained Gaussian processes and the SPDE approach}\label{sec:constrainedGP}
Gaussian processes and random fields are typically specified in terms of their mean and covariance functions. However, a problem with any covariance-based Gaussian model is the computational cost for inference and simulation. Several authors have proposed solutions to this problem, and one particularly important solution is the  GMRF approximation by \cite{lindgren11}. This method is applicable to Gaussian process and random fields with Mat\'ern covariance functions, 
$$
r(h) = \frac{\sigma^2}{\Gamma(\nu)2^{\nu-1}}(\kappa h)^\nu K_\nu(\kappa h), \qquad h\geq 0,
$$
which is the most popular covariance model is spatial statistics, inverse problems and machine learning \cite{Guttorp2006,Rasmussen2006}. The method relies on the fact that a Gaussian random field $X(s)$ on $\mathbb{R}^d$ with a Mat\'ern covariance function can be represented as a solution to the SPDE 
\begin{equation}\label{eq:spde}
	(\kappa^2 - \Delta)^{\frac{\alpha}{2}} X = \phi \mathcal{W},
\end{equation}
where the exponent $\alpha$ is related to $\nu$ via the relation $\alpha = \nu + d/2$, $\Delta$ is the Laplacian, $\mathcal{W}$ is Gaussian white noise on $\mathbb{R}^d$, and $\phi$ is a constant that controls the variance of $X$. The GMRF approximation by \cite{lindgren11} is based on restricting  \eqref{eq:spde} to a bounded domain $\mathcal{D}$, imposing homogeneous Neumann boundary conditions on the operator, and approximating the solution via a finite element method (FEM). The resulting approximation is 
$X_h(s) = \sum_{i=1}^n X_i \varphi_i(s)$,
where $\{\varphi_i(s)\}$ are piecewise linear basis functions induced by a triangulation of the domain, and the vector $\mv{X}$ with all weights $X_i$ is a centered multivariate Gaussian distribution. This can be done for any $\alpha>d/2$ \cite{BK2020rational}, but the case $\alpha\in\mathbb{N}$ is of particular importance since $\mv{X}$ then is a GMRF. In particular, when $\alpha=2$ and $\phi=1$, the precision matrix of $\mv{X}$ is 
$
\mv{Q} = (\kappa^2\mv{C} + \mv{G})\mv{C}^{-1}(\kappa^2\mv{C} + \mv{G}),
$
where $\mv{C}$ is a diagonal matrix with diagonal elements $C_{ii} = \int \varphi_i(s) ds$, and $\mv{G}$ is a sparse matrix with elements $G_{ij} = \int \varphi_i(s)\varphi_j(s) ds$. 

Clearly, a linear constraint on $X_h(s)$ can be written as a constraint on $\mv{X}$. 
For example, if $X_h(s)$ is observed at a location in a given triangle, it creates a linear constraint on the three variables in $\mv{X}$ corresponding to the corners of the triangle. 
Thus, if we draw some observation locations $s_1,\ldots,s_k$ in the domain, we can write 
$\mv{Y} = (X_h(s_1),\ldots, X_h(s_k))^\top = \mv{A}_Y\mv{X}$ where $\mv{A}_Y$ is a $k\times n$ matrix with $\left(\mv{A}_Y\right)_{ij} = \varphi_j(s_i)$. Thus, a model where a Gaussian Mat\'ern fields is observed without measurement noise can efficiently be handled by combining the SPDE approach with the methods from Section~\ref{sec:new_method}. The next section contains a simulation study that compares this combined approach with a standard covariance-based approach in terms of computational cost. 

Through the nested SPDE approach in \cite{bolin11} one can also construct computationally efficient representations of differentiated Gaussian Mat\'ern fields like 
$U(\mv{s}) = (\mv{v}^\top\nabla) X(\mv{s})$, where $\mv{v}^\top\nabla$ is the directional derivative in the direction given by the vector $\mv{v}$ and $X(\mv{s})$ is a sufficiently differentiable Mat\'ern field. A FEM approximation of this model can be written as $U_h(s) = \sum_{i=1}^n U_i \varphi_i(s)$ where now $\mv{U}\sim \pN(\mv{0},\mv{A}_U\mv{Q}^{-1}\mv{A}_U^\top )$. Here $\mv{A}_U$ is a sparse matrix representing the directional derivative and $\mv{Q}$ is the precision matrix of the GMRF representation of $X(\mv{s})$  \cite{bolin11}. If we introduce $\mv{X} \sim \pN(\mv{0},\mv{Q}^{-1})$, we may write $\mv{U} = \mv{A}_U\mv{X}$, and we can thus enforce a restriction on the directional derivative of $\mv{X}$ as a linear restriction $\mv{A}_U\mv{X} = \mv{b}$. As an example, $\mv{v} = (1,1)^\top$ and $\mv{b}=\mv{0}$  results in the restriction $\frac{\pd}{\pd s_1}X(\mv{s}) + \frac{\pd}{\pd s_2}X(\mv{s}) =0$, or in other words that the field is divergence-free. In the next section we use this in combination with the methods in Section~\ref{sec:extension} to construct a computationally efficient Gaussian process regression under linear constraints.


\section{Numerical illustrations} \label{sec:illustrations}
In this section we present two applications. In both cases, timings are obtained through R \cite{Rref} implementations (see the supplementary materials) run on an iMac Pro computer with a 3.2 GHz Intel Xeon processor. The supplementary material contains the source code, including recipes for constructing the figures. 

\subsection{Observations as hard constraints}
Suppose that we have an SPDE approximation $X_h(s)$ of a Gaussian Mat\'ern field $X(s)$, as described above with $\mathcal{D} = [0,1]^2$ and a  triangulation for the GMRF approximation that is based on a uniform mesh with $100 \times 100$ nodes in $\mathcal{D}$. We consider the costs of sampling  $X_h(s)$ conditionally on $k$ point observations without measurement noise, and of log-likelihood evaluations for  these observations. In both cases, the observations are simulated using the parameters $\kappa^2 = 0.5, \phi = 1$ and $\alpha=2$.

We show in the left panel of Figure~\ref{fig:spde} the computation time for evaluating the log-likelihood using the standard method from Section~\ref{sec:old_methods} on the GMRF of weights $\mv{X}$ for the basis expansion of $X_h(\mv{s})$. The panel also shows the corresponding computation time for the new method from Section~\ref{sec:new_method}. The computation times are evaluated for different values of $k$, where for each $k$ the observation locations are sampled uniformly over triangles, and uniformly within triangles, under the restriction that there can be only one observation per triangle, which guarantees that  $\mv{A}_Y\mv{Q}^{-1}\mv{A}_Y^\top$ has full rank. In each iteration, the values of $\kappa^2$ and $\phi$ that are evaluated in the likelihood are sampled from a uniform distribution on $[1,2]$. The curves shown in the figure are computed as averages of $10$ repetitions for each value of $k$.
As a benchmark, we also show the time it takes to evaluate  the log-likelihood assuming that $X(s)$ is a Gaussian Mat\'ern field, which means that we evaluate the log-likelihood of a $k$-dimensional $\pN(\mv{0},\mv{\Sigma})$ distribution without using any sparsity properties. 
\begin{figure}[t]
	\centering
	\includegraphics[width=0.49\linewidth]{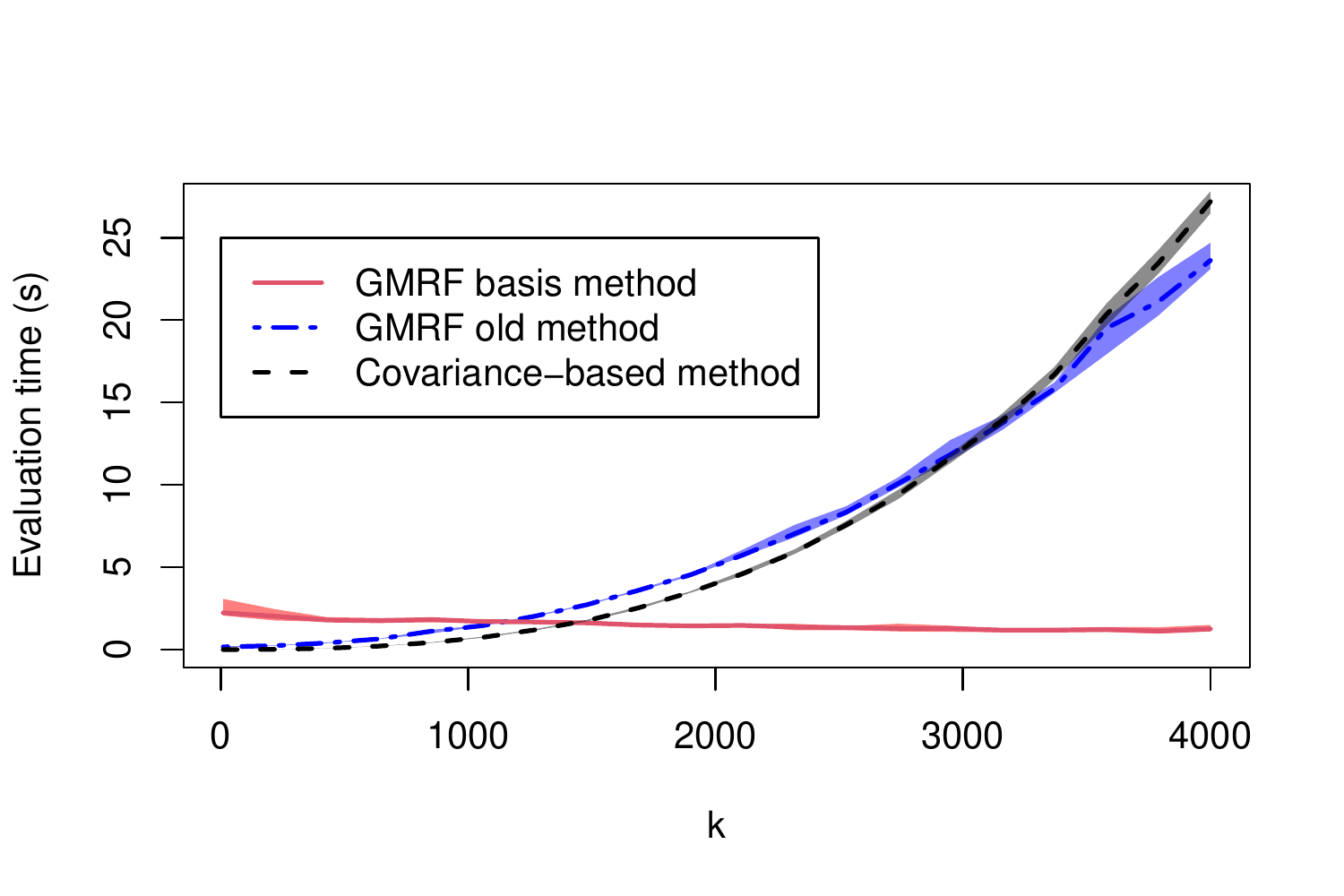}
	\includegraphics[width=0.49\linewidth]{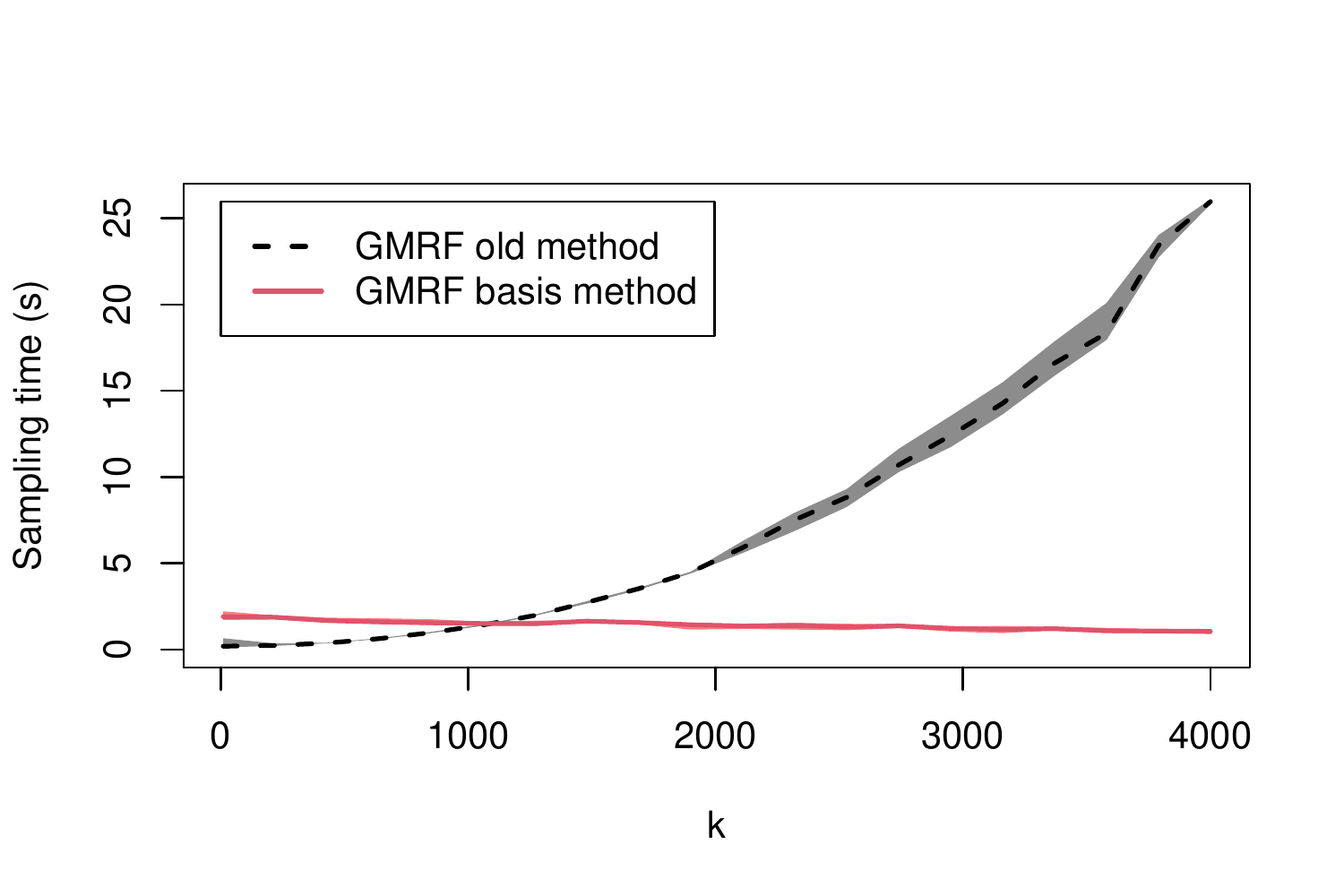}
	\vspace{-0.5cm}
	\caption{Average computation times for one likelihood evaluation (left) and one sample from $\mv{X} | \mv{A}_Y\mv{X} = \mv{y}$ (right) of the Mat\'ern model as a function of the number of observations $k$. As an indication of the uncertainty, the envelopes show the smallest and largest value for each $k$.}
	\label{fig:spde}
\end{figure}

Note that the covariance-based method is the fastest up to approximately 1000 observations, since the problem then is too small for sparsity to be beneficial. For more than 1000 observations, the new method wins and it has in fact a computation time that is decreasing in the number of observations. It should be noted that the difference between the new and old method would be even larger if we reported the computation time for more than one likelihood evaluation, since the construction of the basis needed for the new method only has to be done once. 

In the right panel of Figure~\ref{fig:spde} we show the time needed to sample $X_h(\mv{s})$ conditionally on the observations $\mv{A}_Y \mv{X}=\mv{y}$, i.e., to simulate from $\mv{X}|\mv{A}_Y \mv{X}=\mv{y}$. Both the old method (conditioning by kriging) and the new method (using \eqref{eq:sparsesample} with mean and precision from Theorem~\ref{thm:piXAX}) are shown. We do not show the covariance-based method since it is much slower than both GMRF methods. Also here the displayed values are averages of $10$ repetitions for each value of $k$, and for each repetition the simulation is performed using values  of $\kappa^2$ and $\phi$ that sampled from a uniform distribution on $[1,2]$. The results are similar to those for likelihood evaluations, where the old method is fastest for a low number of observations whereas the new method is much faster for large numbers of observations.

\subsection{Gaussian process regression with  linear constraints}\label{sec:divergence}
We now consider the application from \cite{NIPS2017_6721}, where we assume that we are given noisy observations 
$\mv{Y}_i = \mv{f}(\mv{s}_i) + \mv{\vep}_i$, with $\mv{\vep}_i \sim \pN(0,\sigma^2_e \mv{I})$ of a bivariate  function $\mv{f} = (f_1,\,f_2) : \mathbb{R}^2 \rightarrow \mathbb{R}^2$ with $f_1(\mv{s}) = e^{-as_1s_2}(as_1\sin(s_1s_2) - s_1\cos(s_1s_2))$ and 
$f_2(\mv{s}) = e^{-as_1s_2}(s_2\sin(s_1s_2) - as_2\sin(s_1s_2))$. 
The goal is to use Gaussian process regression to reconstruct $\mv{f}$, under the assumption that we know that it is divergence-free, i.e., $\frac{\pd}{\pd s_1}\mv{f} + \frac{\pd}{\pd s_2}\mv{f} = 0$. We thus want to improve the regression estimate by incorporating this information in the Gaussian process prior for $\mv{f}$. This can be done as in \cite{Scheuerer2012,Wahlstrom2013,NIPS2017_6721} by encoding the information directly in the covariance function, or by imposing the restriction through a hard constraint at each spatial location through the nested SPDE approach. This is done by setting $\mv{B}=  \mv{A}_Y$ and $\mv{A}= \mv{A}_U$ in \eqref{eq:hierarchical}  where the matrices are defined in Section \ref{sec:constrainedGP}.

We choose $a=0.01$ and $\sigma_2 = 10^{-4}$ and generate $50$ observations at randomly selected locations in $[0,4]\times [0,4]$ and predict the function $\mv{f}$ at $N^2 = 20^2$ regularly spaced locations in the square. Independent Mat\'ern priors with $\alpha=4$ are assumed for $f_1$ and $f_2$ and the  covariance-based approach by  \cite{NIPS2017_6721} is taken as a baseline method. As an alternative, we consider the SPDE approximation of the Mat\'ern priors, with $n$ basis functions obtained from a regular triangulation of an extended domain $[-2,6]\times [-2,6]$ (the  extension is added to reduce boundary effects). To be able to use Algorithm~\ref{alg:2}, we only enforce the divergence constraint at every third node for the SPDE model.  This procedure can be seen as an  approximation of the divergence operator that will converge to the true operator when the number of basis functions increases.

The parameters of the baseline model and of the SPDE model are estimated using maximum likelihood, where the likelihood for the SPDE model is computed using Theorem~\ref{Them:piAXsoft}. The function $\mv{f}$ is then reconstructed using the posterior mean of the Gaussian process given the data, which is calculated using Theorem~\ref{Them:piXgby}. This experiment is repeated for $50$ randomly generated datasets.  
In the left panel of Figure~\ref{fig:divergence} we show the average root mean squared error (RMSE) for the reconstruction of $\mv{f}$ for the SPDE model as a function of $n$, based on these $50$ simulations, together with the corresponding RMSE of the baseline method. The shaded region for the SPDE model is a pointwise $95\%$ confidence band.  
One can see that the SPDE model gives a comparable RMSE as long as $n$ is large enough. 

\begin{figure}[t]	
	\centering
	\includegraphics[width=0.49\linewidth]{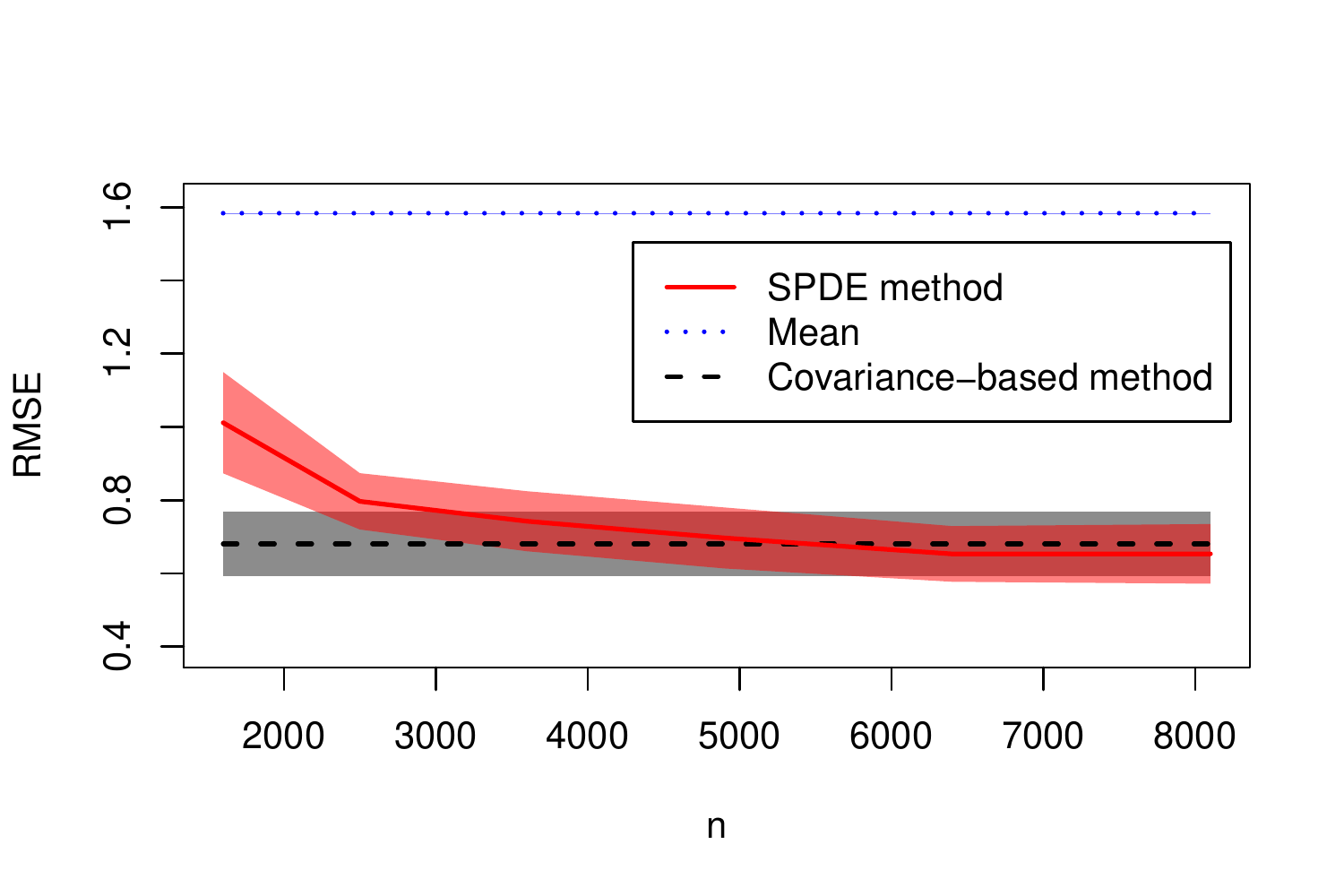}
	\raisebox{0cm}{\includegraphics[width=0.49\linewidth]{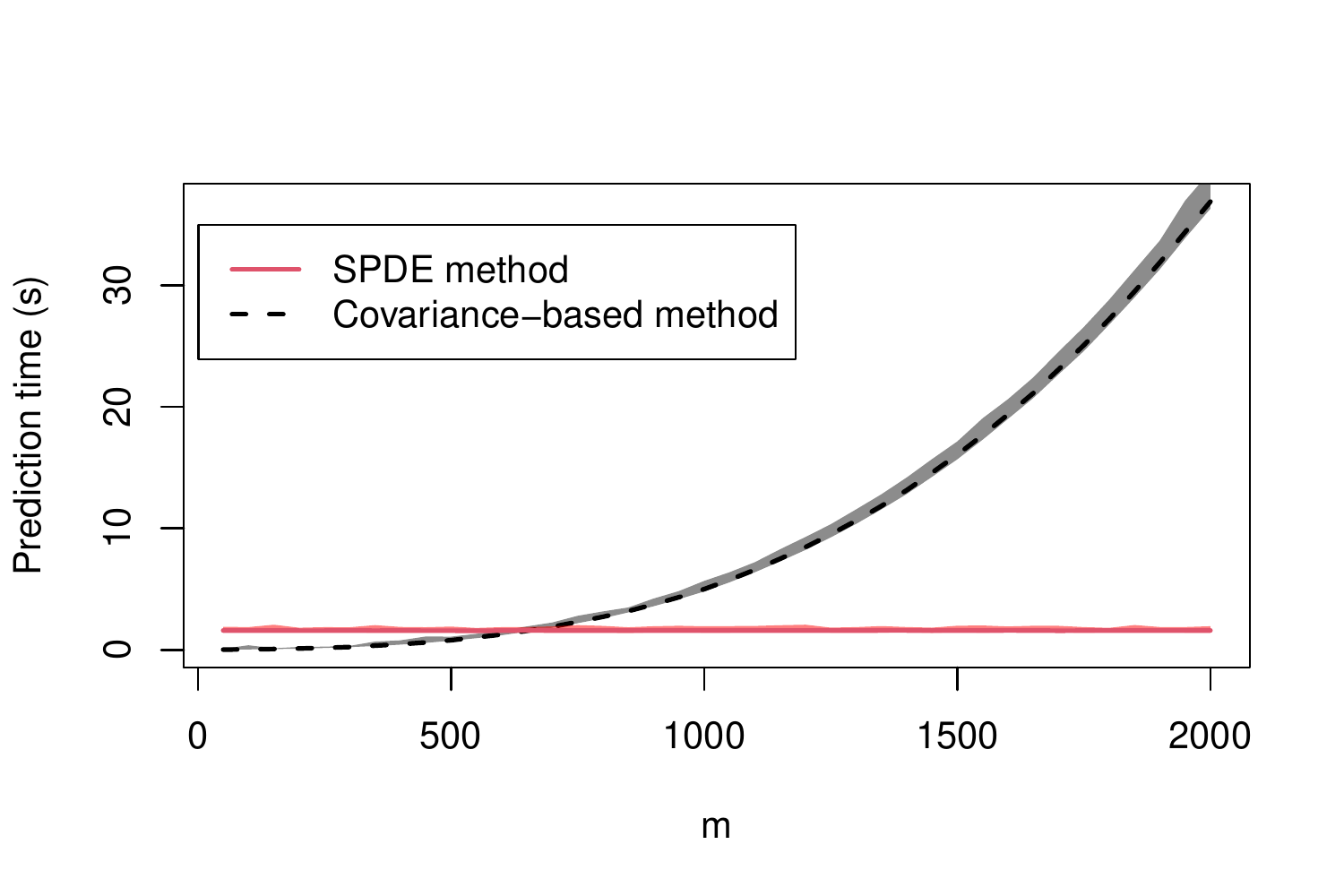}}
	\vspace{-0.5cm}
	\caption{Left: Average RMSE with corresponding $95\%$ pointwise confidence band for $50$ reconstructions of $\mv{f}$ based on SPDE method with different number of basis functions $n$ (red), the corresponding RMSE for the Gaussian process model (black) and the RMSE for estimating $\mv{f}$ by a constant equal to the mean of the observations (blue). Right: Computation time for the prediction of $\mv{f}$ as a function of the number of observations ($m$) with envelopes as in Figure~\ref{fig:spde}.}
	\label{fig:divergence}
\end{figure}

We next fix $n = 3600$ and consider the time it takes to compute a prediction of $\mv{f}$ given the estimated parameters. In the right panel of Figure~\ref{fig:divergence} we show this computation time as a function of the number of observations, $m$, for the baseline method and for the SPDE-based method. Also here we see that the covariance-based method is the fastest for small numbers of observations, whereas the GMRF method (that has a computational cost that scales with the number of basis functions of the SPDE approximation rather than with the number of observations) is fastest whenever $m>600$.

\section{Discussion}\label{sec:conclusions}
We proposed new methods for GMRFs under hard and soft linear constraints, which can greatly reduce computational costs for models with a large number of constraints.
In addition, we showed how to combine these methods with the SPDE approach to allow for computationally efficient linearly constrained Gaussian process regression.	 

Clearly the proposed methods will not be beneficial if the number of constraints is small. Another limitation is that the methods are only efficient if the constraints are sparse. An interesting topic for future research is to handle problems where both sparse and dense constraints are included. In that case one could combine the proposed method with a conditioning by kriging approach where the dense constraints are handled in a post-processing step as described in Section~\ref{sec:old_methods}.

Some recent papers on constrained Gaussian processes, such as \cite{Cong2017hyper}, consider methods that are similar in spirit to those we have developed here. However, to the best of our knowledge, the methods proposed here are the first to account for sparsity, which is the crucial property for GMRFs. We have only considered exact methods in this work, but if one is willing to relax this requirement, an interesting alternative is the Krylov subspace methods by \cite{Simpson2006}. Comparing, or combining, the proposed methods with iterative Krylov subspace methods is an interesting topic for future research. 

In Section \ref{sec:divergence} we showed the advantages of using our proposed method together with SPDE approach for Gaussian process regression with linear constraints. The SPDE approach also allows for more flexible non-stationary covariance structures like the generalized Whittle--Mat\'ern models \cite{BK2020rational}. Our proposed methods are directly applicable to these models in the Markov case (with integer smoothness), and an interesting research direction would be to extend our methods to the case with general smoothness by combining the constraint basis with the rational approximations in \cite{BK2020rational}.

Finally, we can think of no potential negative societal
impacts that this work may have, given that it is solely concerned with
improving the performance of existing methods.

\appendix

\section{Details of the constraint basis construction}\label{sec:details}
In this section we provide more details about the algorithms for constructing the constraint basis. 

A first natural question is why the singular value decomposition (SVD) is a natural method for building a basis with non-interacting hard constraints. To answer that, note that if $\mv{USV}^\top = SVD(\mv{A})$, then by construction the basis $\mv{V}^\top$ is orthonormal and the first $k$ rows span the image of $\mv{A}$ and the last $n-k$ rows span the null-space of $\mv{A}$.  Now, if we let $\mv{x}^*$ denote a vector $\mv{x}$ expressed in the basis $\mv{V}^\top$, then $\mv{x}^*$ can be transformed back to the natural basis by $\mv{x}=\mv{V}\mv{x}^*$ hence 
$$
\mv{A}\mv{x}= \mv{b} \quad \Leftrightarrow\quad  \mv{US} \mv{x}^*=\mv{b} \quad\Leftrightarrow\quad \begin{bmatrix}
	\mv{U}\mv{S}_{\A,\A} & \mv{0} 
\end{bmatrix} \mv{x}^*=\mv{b} 
\quad\Leftrightarrow\quad 
\mv{U}\mv{S}_{\A,\A}  \mv{x}^*_{\A} = \mv{b}.
$$ 

A second question that should be addressed is how the reordering of the $\mv{A}$ matrix in Algorithm~\ref{alg:2} is done. This is illustrated in Algorithm~\ref{alg:tildeVec} where we show how to build the sub-matrices $\{\tilde{\mv{A}}\}_{k=1}^m$.

\newcommand{\Break}{\State \textbf{break} }
\begin{figure}[H]
	\begin{center}
		\begin{minipage}[t]{0.45\linewidth}
			\begin{algorithm}[H]
				\caption{Find all non-overlapping sub-matrices}
				\label{alg:tildeVec}
				\begin{algorithmic}[1]
					\Require $\mv{A}$ (a $k \times n$ matrix)
					\State $\{\tilde{\mv{A}}_1,\mv{B} \} \gets overlap(\mv{A})$
					\State $m\gets 1$
					\While{ $\mv{B} \neq \emptyset$}
					\State $m\gets m+1$
					\State $\{\tilde{\mv{A}}_m,\mv{B} \} \gets overlap(\mv{B})$
					\EndWhile
					\State Return $\{\tilde{\mv{A}}\}_{k=1}^m$
				\end{algorithmic}
			\end{algorithm}	
		\end{minipage}
		\hspace{0.2cm}
		\begin{minipage}[t]{0.49\linewidth}
			\begin{algorithm}[H]
				\caption{$overlap(\mv{A})$Find first sub-matrix}
				\label{alg:tilde}
				\begin{algorithmic}[1]
					\Require $\mv{A}$ (a $k \times n$ matrix)
					\State $U \gets \{1\}$
					\State $d \gets 0$
					\State $D \gets  id(\mv{A}_{U,\mv{\cdot}})$    
					\While{ $1$}
					\State $D \gets  id(\mv{A}_{U,\mv{\cdot}})$     
					\State $U \gets id\left(\left(\mv{A}_{\mv{\cdot},D} \right)^\top \right)$ 
					\If{$d=|U|$}
					\Break
					\EndIf
					\State $d \gets  |U|$
					\EndWhile
					\State $\widetilde{\mv{A}} = \mv{A}_{\mv{\cdot},U}$
					\State $\widetilde{\mv{A}}^c = \mv{A}_{\mv{\cdot},U^c}$
					\State Return $\{\widetilde{\mv{A}}, \widetilde{\mv{A}}^c\}$
				\end{algorithmic}
			\end{algorithm}	
		\end{minipage}
	\end{center}
\end{figure}
\section{Proofs}\label{app:proof}
In this section we prove the four main theorems of the paper. 
\begin{proof}[Proof of Theorem \ref{Them:piAX}]
	We first transform the density of $\mv{AX}$ to the basis represented by $\mv{T}$,
	\begin{align*}
		\pi_{\mv{AX}}(\mv{b}) = \pi_{\mv{AT}^\top\mv{X}^*}(\mv{b})=\pi_{\mv{X}_\A^*}\left(\mv{H}^{-1}\mv{b }\right
		) \left| |\mv{H}|^{-1}  \right|
		= \pi_{\mv{X}_\A^*}\left(\mv{H}^{-1}\mv{b }\right
		) |\mv{A}\mv{A}^\top|^{-1/2}.
	\end{align*} 
	In order to derive the density $\mv{X}_\A^*$, note that the density of $\mv{X}^*$ is 
	$$
	\pi_{\mv{X}_\A^*}(\mv{x}^*) = \frac{|\mv{Q}|^{\dagger/2}}{(2\pi)^{\frac{n-s}{2}}}\exp\left(-\frac1{2}Q(\mv{x}^*)\right),
	$$ 
	where the quadratic form $Q(\mv{x}^*)$ is
	\begin{align*}
		Q(\mv{x}^*)  =&  \begin{bmatrix}
			\mv{x}^*_{\Ac} - \mv{\mu}_{\Ac}^*\\
			\mv{x}^*_{\A} - \mv{\mu}_{\A}^*
		\end{bmatrix}^\top
		\begin{bmatrix}
			\mv{Q}_{\Ac\Ac}^{*} & \mv{Q}_{\A\Ac}^{*} \\
			\mv{Q}_{\A\Ac}^{*} & \mv{Q}_{\A\A}^{*}
		\end{bmatrix} 
		\begin{bmatrix}
			\mv{x}^*_{\Ac} - \mv{\mu}_{\Ac}^*\\
			\mv{x}^*_{\A} - \mv{\mu}_{\A}^*
		\end{bmatrix} \\
		=&  	
		\left(\mv{x}^{*}_{\A}- \mv{\mu}_{\A}^*\right)^\top 	
		\mv{Q}_{\A\A}^{*}\left(\mv{x}^*_{\A}  - \mv{\mu}_{\A}^*\right)  +		
		\left(\mv{x}^*_{\A} -\mv{\mu}_{\A}^*\right)^\top
		\mv{Q}_{\A\Ac}^{*}\left(\mv{x}^*_{\Ac}-\mv{\mu}_{\Ac}^*\right) \\
		&+ 
		\left(\mv{x}^*_{\Ac}-\mv{\mu}_{\Ac}^*\right)^\top \mv{Q}_{\Ac\A}^{*} \left(\mv{x}^*_{\A}-\mv{\mu}_{\A}^*\right) 
		+ \left( \mv{x}^*_{\Ac} - \mv{\mu}_{\Ac}^*\right)^\top  \mv{Q}_{\Ac\Ac}^{*}
		\left( \mv{x}^*_{\Ac} - \mv{\mu}_{\Ac}^* \right) \\
		=&  	(\mv{x}^{*}_{\A}- \mv{\mu}_{\A}^*)^\top 	\mv{Q}_{\A\A}^{*}	(\mv{x}^*_{\A}  - \mv{\mu}_{\A}^*) -	\left(\mv{x}^*_{\A}-\mv{\mu}_{\A}^*\right)^\top \mv{Q}_{\A\Ac}^{*}\mv{Q}_{\Ac\Ac}^{*\dagger}\mv{Q}_{\Ac\A}^{*}\left(\mv{x}^*_{\A} - \mv{\mu}_{\A}^*\right) 	+ \\  
		&\hspace{-1cm}+\left( \mv{x}^*_{\Ac} - \mv{\mu}_{\Ac}^* + \mv{Q}_{\Ac\Ac}^{*\dagger} \mv{Q}_{\Ac\A}^{*} (\mv{x}^*_{\A} - \mv{\mu}_{\A}^*) \right)^\top  \mv{Q}_{\Ac\Ac}^{*}
		\left( \mv{x}^*_{\Ac} - \mv{\mu}_{\Ac}^*
		+ \mv{Q}_{\Ac\Ac}^{*\dagger} \mv{Q}_{\Ac\A}^{*} (\mv{x}^*_{\A} - \mv{\mu}_{\A}^*) \right).
	\end{align*}
	Here we in the  last step wrote the expression so that we easily can integrate out $\mv{X}^*_{\Ac}$ on the complement to the null space of $\mv{Q}^*_{\Ac\Ac}$. Doing so yields the desired result,
	\begin{align*}
		\pi(\mv{x}_{\A}^*) =& \int 	\pi_{\mv{X}_\A^*}(\mv{x}^*) d\mv{x}^*_{\Ac}  
		\propto  \frac{|\mv{Q}|^{\dagger/2}}{|\mv{Q}^*_{\Ac\Ac}|^{\dagger/2}} 
		\exp \left(- \frac{1}{2}		\hat{Q}(\mv{x}_{\A}^*) \right),
	\end{align*}
	where 
	$$
	\hat{Q}(\mv{x}_{\A}^*) =  \left(	\mv{x}^*_{\A} - \mv{\mu}^*_{\A} \right)^\top 	\left(\mv{Q}_{\A\A}^{*} - \mv{Q}_{\A\Ac}^{*}\mv{Q}_{\Ac\Ac}^{*\dagger}\mv{Q}_{\Ac\A}^{*}\right)\left(\mv{x}^*_{\A}-\mv{\mu}^*_{\A} \right). 
	$$
\end{proof}

To prove Theorem~\ref{thm:piXAX} we need the following lemma.
\begin{Lemma}
	\label{lem:rank}	
	Under Assumption \ref{ass:main} one has $rank\left(\mv{Q}_{\A\A}^* \right)=k-k_0$ and $rank \left(\mv{Q}_{\Ac\Ac}^* \right)=n-s-(k-k_0)$.
\end{Lemma} 
\begin{proof}
	We have $rank(\mv{Q}^*) = rank(\mv{Q}) =n-s$ since $\mv{T}$ is orthonormal matrix. Further, using the eigen-decomposition of $\mv{Q}$ we can express $\mv{Q}^*$ as
	$$
	\mv{Q}^* = \begin{bmatrix}
		\mv{T}_{\A} \\
		\mv{T}_{\Ac}
	\end{bmatrix} 
	\begin{bmatrix}
		\mv{E}_{0^c} \\
		\mv{E}_{0}
	\end{bmatrix} 
	\begin{bmatrix}
		\mv{\Lambda} & \mv{0}\\
		\mv{0} & \mv{0}
	\end{bmatrix} 
	\begin{bmatrix}
		\mv{E}_{0^c} \\
		\mv{E}_{0}
	\end{bmatrix}^\top
	\begin{bmatrix}
		\mv{T}_{\A} \\
		\mv{T}_{\Ac}
	\end{bmatrix}^\top, 
	$$
	where $\mv{\Lambda}$ is a diagonal matrix with the non-zero eigenvalues of $\mv{Q}$. Since $rank(\mv{A}\mv{E}_0)=k_0$ it follows that also $rank(\mv{T}_{\A}\mv{E}_0)=k_0$ and $rank(\mv{T}_{\Ac}\mv{E}_0)=s-k_0$. 
	By Theorem 4.3.28 of \cite{horn2012matrix} there exists an eigenvector, $\mv{e}$, of $\mv{Q}_{\Ac\Ac}^{*}$ that has a corresponding eigenvalue $0$ if and only if $\mv{Q}_{\Ac\Ac}^{*}\mv{e}=\mv{0}$ and $\mv{Q}_{\A\Ac}^{*}\mv{e}=\mv{0}$. 
	By construction, any vector constructed by the linear span of $\mv{T}_{\Ac}\mv{E}_0$ satisfies this requirement, and no other vector does. Hence, the rank of $\mv{Q}_{\Ac\Ac}^*$ is $n-k-(s-k_0)$ and the rank of  $\mv{Q}^*_{\A\A}$ is $k-k_0$.
\end{proof}

\begin{proof}[Proof of Theorem \ref{thm:piXAX}]
	To derive the distribution we note that the conditional distribution  of  $\mv{X}_{\Ac}^*|\mv{X}_\A^*$  is proportional to $\exp(-\frac1{2}Q(\mv{x}^*))$, where 
	\begin{align*}
		Q(\mv{x}^*) &= \begin{bmatrix}
			\mv{x}^*_{\Ac} - \mv{\mu}^*_{\Ac} \\
			\mv{x}^*_{\A} - \mv{\mu}^*_{\A}
		\end{bmatrix}^\top
		\begin{bmatrix}
			\mv{Q}_{\Ac\Ac}^{*} & \mv{Q}_{\A\Ac}^{*} \\
			\mv{Q}_{\A\Ac}^{*} & \mv{Q}_{\A\A}^{*}
		\end{bmatrix} 
		\begin{bmatrix}
			\mv{x}^*_{\Ac} - \mv{\mu}^*_{\Ac}\\
			\mv{x}^*_{\A} - \mv{\mu}^*_{\A}
		\end{bmatrix} \\
		&=
		\left( \mv{x}^{*}_{\Ac} - \mv{\mu}^*_{\Ac}\right)^\top 	\mv{Q}_{\Ac\Ac}^{*} \left(	\mv{x}^*_{\Ac} - \mv{\mu}^*_{\Ac}\right) + 2  \left(\mv{x}^{*}_{\Ac}-\mv{\mu}^*_{\Ac}\right)^\top \mv{Q}_{\Ac\A}^{*}  \left(\mv{x}^*_{\A} -\mv{\mu}^*_{\A}\right) + C,
	\end{align*}
	where $C$ is a constant independent of $\mv{x}_{\Ac}^*$. Now, since $\mv{Q}_{\Ac\Ac}^* \left( \mv{Q}_{\Ac\Ac}^{*} \right)^{\dagger} \mv{Q}^{*}_{\Ac\A}= \mv{Q}^{*}_{\Ac\A}$, the quadratic form $Q(\mv{x}^*)$ can be written as a constant plus $\mv{v}^\top\mv{Q}_{\Ac\A}^{*} \mv{v}$, where 
	$$
	\mv{v} =  \mv{x}^*_{\Ac} - \mv{\mu}^*_{\Ac}+\mv{Q}_{\Ac\Ac}^{*\dagger} \mv{Q}_{\Ac\A}^{*} (\mv{x}^*_{\A}- \mv{\mu}^*_{\A}).
	$$
	Hence 
	\begin{equation}\label{eq:Xconditional}
		\mv{X}_{\Ac}^*|\mv{X}^*_\A=\mv{b}^* \sim  \mathcal{N}_C\left(\mv{Q}_{\Ac\Ac}^*\left(
		\mv{\mu}^*_{\Ac}-\mv{Q}_{\Ac\Ac}^{*\dagger} \mv{Q}_{\Ac\A}^{*} \left( \mv{b}^* -  \mv{\mu}^*_{\A}\right)\right)
		, \mv{Q}_{\Ac\Ac}^* \right).
	\end{equation}
	Combining this with the fact that $\mv{X}=\mv{T}^\top \mv{X}^*$ gives the desired expression for the distribution of $\mv{X}|\mv{AX}=\mv{b}$. Finally, since $\mv{T}$ is orthonormal, the rank of $\mv{Q}_{X|b}$ is the same as the rank of $\mv{Q}_{\Ac\Ac}$, and the result follows from Lemma \ref{lem:rank}.
\end{proof}

We start by proving Theorem~\ref{Them:piXgby} as we will use result from that proof in  the proof of Theorem \ref{Them:piAXsoft}.

\begin{proof}[Proof of Theorem~\ref{Them:piXgby}]
	First, note that the density 
	\begin{align}\label{eq:piYcond}
		\pi_{\mv{Y}|\mv{X}_\Ac^*,\mv{X}_{\A}^*}(\mv{y}|\mv{x}_{\Ac}^*,\mv{b}^*) &= \frac{1}{(2\pi)^{\frac{m}{2}}\sigma_Y^{m}}\exp \left( -\frac{1}{2\sigma_Y^2}\left( \mv{y}- \mv{B}^*
		\begin{bmatrix}
			\mv{b}^* \\ 
			\mv{x}^*_{\Ac} 
		\end{bmatrix}
		\right)
		\left( \mv{y}- \mv{B}^*
		\begin{bmatrix}
			\mv{b}^* \\ 
			\mv{x}^*_{\Ac} 
		\end{bmatrix}
		\right)\right),
	\end{align}
	can, as a function of $\mv{x}_{\Ac}^*$, be written as
	\begin{align*}
		\pi_{\mv{Y}|\mv{X}_\A^*,\mv{X}_{\Ac}^*}(\mv{y}|\mv{b}^*, \mv{x}_{\Ac}^* ) 
		\propto& \exp\left (-\frac{\mv{x}^{*\top}_{\Ac}\mv{B}^{*\top}_{\Ac} \mv{B}^*_{\Ac}\mv{x}^*_{\Ac}}{2\sigma_Y^2} + \frac{\mv{y}^{*\top}\mv{B}^*_{\Ac}\mv{x}^*_{\Ac}}{\sigma_Y^2} \right).
	\end{align*}
	Further, from \eqref{eq:Xconditional}, we have that, as a function of $\mv{x}_{\Ac}^*$,  
	\begin{align*}
		&\pi_{\mv{X}^*_{\Ac}|\mv{X}^*_\A}(\mv{x}^*_{\Ac}|\mv{b}^*) \propto
		\exp \left( -\frac{1}{2}  \left( \mv{x}^*_{\Ac}- \widetilde{\mv{\mu}}^*_{\Ac} \right)^\top  \mv{Q}^*_{\Ac\Ac}  \left( \mv{x}^*_{\Ac}- \widetilde{\mv{\mu}}^*_{\Ac} \right) \right),
	\end{align*}
	where $\widetilde{\mv{\mu}}^*_{\Ac} =  \mv{\mu}^*_{\Ac}-\mv{Q}_{\Ac\Ac}^{*\dagger} \mv{Q}_{\Ac\A}^{*} \left( \mv{b}^* -  \mv{\mu}^*_{\A}\right)$.
	Since $\pi_{\mv{X}_{\Ac}^*|\mv{Y},\mv{X}^*_{\A} } (\mv{x}^*_{\Ac}| \mv{y},\mv{b}^*)$ is proportional  to 
	$
	\pi_{\mv{Y}|\mv{X}_\A^*,\mv{X}_{\Ac}^*}(\mv{y}|\mv{b}^*, \mv{x}_{\Ac}^* ) \pi(\mv{x}^*_{\Ac}|\mv{b}^*),
	$
	it follows that
	\begin{align*}
		\pi_{\mv{X}_{\Ac}^*|\mv{Y},\mv{X}^*_{\A} } (\mv{x}^*_{\Ac}| \mv{y},\mv{b}^*)  \propto
		&\exp\left (-\frac{1}{2}\mv{x}^{*\top}_{\Ac} \frac{\mv{B}^{*\top}_{\Ac} \mv{B}^*_{\Ac}}{\sigma_Y^2} \mv{x}^*_{\Ac}+ \left( \frac{\mv{B}^{*\top}_{\Ac} \mv{y}^{*}}{\sigma_Y^2} \right)^{\top}\mv{x}^*_{\Ac} \right) \cdot \\  
		& \exp\left(-\frac{1}{2} \mv{x}^{*\top}_{\Ac} \mv{Q}^*_{\Ac\Ac} \mv{x}^*_{\Ac}+ \left(  \mv{Q}^*_{\Ac\Ac} \widetilde{\mv{\mu}}^*_{\Ac}\right)^{\top} \mv{x}^*_{\Ac}   \right) \\
		\propto 	&\exp\left(- \frac{1}{2} \left(\mv{x}_{\Ac}^*-\widehat{\mv{\mu}}_{\Ac}^* \right)^\top \widehat{\mv{Q}}_{\Ac\Ac}^*
		\left(\mv{x}_{\Ac}^*-\widehat{\mv{\mu}}_{\Ac}^* \right)\right).
	\end{align*}
	Finally, using the relation $\mv{X}= \mv{T}^\top\mv{X}^*$ completes the proof.
\end{proof}

\begin{proof}[Proof of Theorem \ref{Them:piAXsoft}]
	First note that $\pi_{\mv{Y} | \mv{AX}}\left(  \mv{y} | \mv{b} \right)= \pi_{\mv{Y} | \mv{X}_\A^*}(\mv{y}|\mv{b}^*)$ and
	\begin{align}
		\pi_{\mv{Y} | \mv{X}^*_{\A}}\left(  \mv{y} | \mv{b}^* \right) &=  \int  \pi_{\mv{X}^*_{\Ac},\mv{Y}| \mv{X}^*_{\A} }\left(\mv{x}^*_{\Ac}, \mv{y}| \mv{b}^* \right) d\mv{x}^*_{\Ac} \notag\\
		&=  \int  \pi_{\mv{Y}|\mv{X}^*_{\Ac}, \mv{X}^*_{\A}}(\mv{y}|\mv{x}_{\Ac}^*, \mv{b}^*) \pi_{\mv{X}^*_{\Ac}| \mv{X}^*_{\A}}(\mv{x}_{\Ac}^*| \mv{b}^*) d\mv{x}^*_{\Ac}.	\label{eq:decomposePI}
	\end{align}
	The goal is now to derive an explicit form of the density by evaluating the integral in \eqref{eq:decomposePI}. 
	By the expressions in the proof of Theorem~\ref{Them:piXgby} we have
	\begin{align*}
		\pi_{\mv{Y}|\mv{X}^*_{\Ac}, \mv{X}^*_{\A}}(\mv{y}|\mv{x}_{\Ac}^*, \mv{b}^*) \pi_{\mv{X}^*_{\Ac}| \mv{X}^*_{\A}}(\mv{x}_{\Ac}^*| \mv{b}^*) =&  
		\exp\left (-\frac{1}{2}\mv{x}^{*\top}_{\Ac} \frac{\mv{B}^{*\top}_{\Ac} \mv{B}^*_{\Ac}}{\sigma_Y^2} \mv{x}^*_{\Ac}+ \left( \frac{\mv{B}^{*\top}_{\Ac} \mv{y}^{*}}{\sigma_Y^2} \right)^{\top}\mv{x}^*_{\Ac} \right) \cdot \\  
		& \exp\left(-\frac{1}{2} \mv{x}^{*\top}_{\Ac} \mv{Q}^*_{\Ac\Ac} \mv{x}^*_{\Ac}+ \left(  \mv{Q}^*_{\Ac\Ac} \widetilde{\mv{\mu}}^*_{\Ac}\right)^{\top} \mv{x}^*_{\Ac}   \right) \cdot \\
		& \frac{ |\mv{Q}^*_{\Ac\Ac}|^{\dagger/2}}{\left(2\pi\right)^{c_0} \sigma^{m}_Y} \exp \left(- \frac{1}{2}\left[  \frac{\mv{y}^{*\top}\mv{y}^*}{\sigma^2_Y} +  \widetilde{\mv{\mu}}_{\Ac}^{*\top}  \mv{Q}^*_{\Ac\Ac}  \widetilde{\mv{\mu}}^*_{\Ac} \right] \right) \\
		=&\pi_{\mv{X}^*_{\Ac}|\mv{Y}, \mv{X}^*_{\A}}(\mv{x}_{\Ac}^*| \mv{y},\mv{b}^*)  \frac{\exp \left( \frac{1}{2}  \widehat{\mv{\mu}}_{\Ac}^{*\top} \widehat{\mv{Q}}^*_{\Ac\Ac} \widehat{\mv{\mu}}^*_{\Ac} \right) }{|\widehat{\mv{Q}}_{\Ac\Ac}^* |^{\dagger/2}}
		\cdot \\
		& 	\frac{ |\mv{Q}^*_{\Ac\Ac}|^{\dagger/2}}{\left(2\pi\right)^{c_1} \sigma^{m}_Y} \exp \left(- \frac{1}{2}\left[  \frac{\mv{y}^{*\top}\mv{y}^*}{\sigma^2_Y} +  \mv{\mu}_{\Ac}^{*\top}  \mv{Q}^*_{\Ac\Ac}\mv{\mu}^*_{\Ac}\right]\right),
	\end{align*}
	where $c_0$ and $c_1$ are positive constants. Inserting this expression in \eqref{eq:decomposePI} and evaluating the integral, where one notes that   $\pi_{\mv{X}^*_{\Ac}|\mv{Y}, \mv{X}^*_{\A}}(\mv{x}_{\Ac}^*| \mv{y},\mv{b}^*) $ integrates to one, gives the desired result. 
\end{proof}

\section{Conditional constrained distribution}\label{app:conditional}
In order to derive the conditional density $\pi(\mv{x}|\mv{Ax}=\mv{b})$ in \eqref{eq:old_likelihood}  we will use what is known as the disintegration technique. The proof is built on the results in  \cite{chang1997conditioning}, which has the following definition.
\begin{Definition}
	\label{def:dis}
	Let $(\mathcal{X}, \mathcal{A},\lambda)$ and $(\mathcal{T}, \mathcal{B},\mu)$ be two measure spaces with $\sigma$-finite measures $\lambda$ and $\mu$. 
	The measure $\lambda$ has a disintegration $\{\lambda_{b}\}$ with respect to the measurable map $A:(\mathcal{X}, \mathcal{A}) \rightarrow (\mathcal{T}, \mathcal{B})$ and the measure $\mu$, or a $(A(x),\mu)-$disintegration if:
	\begin{itemize}
		\item [(i)] $\lambda_{b}$ is a $\sigma$-finite measure on $\mathcal{A}$ such that $\lambda_{b}\left(A(x)\neq b\right)=0,$ for $\mu-$almost all $b$,
	\end{itemize}
	and, for each non-negative measurable function $f$ on $\mathcal{X}$:
	\begin{itemize}
		\item [(ii)] $b\rightarrow \int f d\lambda_b$ is measurable.
		\item[(iii)] $\int f d\lambda =\int \int  f d\lambda_b d\mu$.
	\end{itemize}	
\end{Definition}

In the following theorem, we use the notation from Appendix \ref{app:proof} and  let $\lambda_n$ denote the Lebesgue measure on $\mathbb{R}^n$. Further, we define $\lambda_{\Ac}$ as the image measure of the projection onto the image of $\mv{A}$ (which is not $\sigma$-finite), and  $\lambda_{\A}$ as the image measure of the projection onto the null-space of $\mv{A}$.

\begin{Theorem}
	\label{thm:AXborig}
	Let $\mv{X}$ be a multivariate random variable with distribution $\mathbb{P}$ on $(\mathbb{R}^n,\mathcal{B}(\mathbb{R}^n))$, where  $\mathbb{P} $ has density $\pi \left( \mv{x}\right)$ with respect to $\lambda_n$ . Then the random variable $\mv{X} | \mv{AX}=\mv{b}$ has density
	$$
	\pi\left(\mv{x}| \mv{Ax}=\mv{b}\right) = \frac{\mathbb{I}\left(\mv{Ax}=\mv{b}\right)|\mv{A}\mv{A}^\top|^{-1/2}\pi(\mv{x})}{\pi_{\mv{AX}}(\mv{b})},
	$$
	with respect to the measure $\mathcal{L}_{\mv{b}}(\cdot)=  \lambda_{\Ac}(\cdot \cap \{x:\mv{Ax}=\mv{b}\})$ on  $(\mathbb{R}^n,\mathcal{B}(\mathbb{R}^n))$.
\end{Theorem}
The proof is based on the following lemma.
\begin{Lemma}\label{lem:disintegration}
	The measure $\mathcal{L}_{\mv{b}}(\cdot)$  is the $(\mv{A},\mathcal{L}^k)$-disintegration of the Lebesgue measure $\lambda_n$.
\end{Lemma}
\begin{proof}
	Thus we need to show that (i), (ii), and (iii) of Definition \ref{def:dis} holds.
	Clearly, (i) follows immediate from  $\cdot \cap \{x:\mv{Ax}=\mv{b}\}$.
	To show (ii), note that 
	\begin{align*}
		\int f d\mathcal{L}_{\mv{b}}  &= \int f\left(\mv{T}^\top\mv{x}^* \right)
		\mathbb{I}_{\mv{A\mv{T}^\top\mv{x}^*}=\mv{b}}\left(d\mv{x}^* \right) d\mv{x}^*_{\Ac} 
		\\ &=\left| |\mv{H}| \right|\int_{\{\mv{x}^*:\mv{x}_\A^*=\mv{H}^{-1}\mv{b}\}} f^*(\mv{x}_{\A}^*,\mv{x}_{\Ac}^*) d\mv{x}^*_{\Ac} 
		= \left| |\mv{H}| \right| \int f^*(\mv{H}^{-1}\mv{b},\mv{x}_{\Ac}^*)d\mv{x}^*_{\Ac},
	\end{align*}
	where  $\mv{H} = \bigl(\mv{A}\mv{T}^\top\bigr)_{\A\A}$ as defined in Section~\ref{sec:sampling_new}, $\left| |\mv{H}| \right|$ denotes the absolute value of the determinant of $\mv{H}$, and 
	$f^*(\mv{x}^*)=f(\mv{T}^\top\mv{x})$. Since $f$ is a measurable function it follows by Tonelli Theorem \cite{pollard2002user} that above partial integral is measurable. Finally, to show (iii), we continue from the equation above and get
	\begin{align*}
		\iint f d\mathcal{L}_{\mv{b}} d\mv{b} 
		&= \int  \left| |\mv{H}| \right|  \int f^*(\mv{H}^{-1}\mv{b},\mv{x}_{\Ac}^*)d\mv{x}^*_{\Ac}d\mv{b}\\
		&=  \left| |\mv{H}| \right|   \left| |\mv{H}|^{-1} \right| \iint  f^*(\mv{b}^*,\mv{x}_{\Ac}^*)d\mv{x}^*_{\Ac}d\mv{b}^* 
		= \int f d\lambda_n.
	\end{align*}
\end{proof}
\begin{proof}[Proof of Theorem~\ref{thm:AXborig}]
	By Lemma \ref{lem:disintegration} above and  Theorem 3 (v) in \cite{chang1997conditioning} it follows that the random variable has density
	\begin{align*}
		\pi(\mv{x}|\mv{AX}=\mv{b}) &= \frac{\pi(\mv{x})}{\mathcal{L}_{\mv{b}}\pi(\mv{x})}=\frac{\pi(\mv{x})}{\pi_{\mv{AX}}(\mv{b})\left||\mv{H}|\right|}  =\frac{|\mv{A}\mv{A}^\top|^{-1/2}\pi(\mv{x})}{\pi_{\mv{AX}}(\mv{b})} ,
	\end{align*}
	a.e. with respect to $\mathcal{L}_{\mv{b}}$. Finally, it holds that $\pi(\mv{x}|\mv{AX}=\mv{b}) = \mathbb{I}\left(\mv{Ax}=\mv{b}\right) \pi(\mv{x}|\mv{AX}=\mv{b})  $ a.e. since $\mathcal{L}_{\mv{b}}(\cdot)=  \lambda_{\Ac}(\cdot \cap \{\mv{x}:\mv{Ax}=\mv{b}\})$.
\end{proof}

\printbibliography

\end{document}